\documentclass[conference]{IEEEtran}
\IEEEoverridecommandlockouts
\usepackage{cite}
\usepackage{amsmath,amssymb,amsfonts}
\usepackage{algorithmic}
\usepackage{graphicx}
\usepackage{textcomp}
\usepackage{xcolor}
\usepackage{amsmath,amsfonts, amssymb, graphicx,amsthm,epsfig}
\usepackage{bm,bbm}  
\usepackage{color}
\usepackage{amssymb}
\usepackage{dsfont}
\usepackage{mathtools}
\usepackage{array,multirow}
\usepackage{tikz-network}
\usepackage{subfig} 

\def\BibTeX{{\rm B\kern-.05em{\sc i\kern-.025em b}\kern-.08em
    T\kern-.1667em\lower.7ex\hbox{E}\kern-.125emX}}

\newcommand{\bxi}{\bm{\xi}}
\newcommand{\bpsi}{\bm{\psi}}
\newcommand{\bmu}{\bm{\mu}}
\newcommand{\bS}{\bm{\lambda}}

\newcommand{\cN}{\mathcal{N}}

\newcommand{\bL}{\bm{x}}

\newcommand{\bA}{\bm{A}}
\newcommand{\cA}{\mathcal{A}}

\newcommand{\bw}{\bm{w}}
\newcommand{\ba}{\bm{a}}
\newcommand{\cF}{\mathcal{F}}
\newcommand{\cE}{\mathcal{E}}

\newcommand{\bK}{\bm{y}}
\newcommand{\bell}{\bm{\ell}}
\newcommand{\bz}{\bm{z}}

\newcommand{\bmm}{\bm{m}}
\newcommand{\E}{\mathbb{E}}
\newcommand{\bbP}{\mathbb{P}}
\newcommand{\bQ}{\mathbb{Q}}

\DeclareMathOperator{\asto}{\stackrel{\text{a.s.}}{\longrightarrow}}
\DeclareMathOperator*{\argmax}{arg\,max}
\newcommand{\iprod}[2]{\left\langle #1,#2 \right\rangle}

\newtheorem{theorem}{Theorem}
\newtheorem{lemma}{Lemma}
\newtheorem{corollary}{Corollary}
\newtheorem{assumption}{Assumption}
    
\begin{document}

\title{Social Learning under Randomized Collaborations}
\author{\IEEEauthorblockN{Yunus \.Inan, Mert Kayaalp, Emre Telatar, and Ali H. Sayed\thanks{This work was supported in part by SNSF grant 205121-184999.}}
		\IEEEauthorblockA{EPFL, Lausanne, Switzerland}}
\maketitle
\begin{abstract}We study a social learning scheme where at every time instant, each agent chooses to receive information from one of its neighbors at random. We show that under this sparser communication scheme, the agents learn the truth eventually and the asymptotic convergence rate remains the same as the standard algorithms, which use more communication resources. We also derive large deviation estimates of the log-belief ratios for a special case where each agent replaces its belief with that of the chosen neighbor.
\end{abstract}

	\section{Introduction}
	Social learning is a paradigm that investigates how opinions are formed over social networks by modeling the interactions between networked agents to infer the true state-of-nature \cite{chamley_2003,acemoglu_2011,krishnamurthy_2013,chamley_2013}. In non-Bayesian social learning \cite{jadbabaie_2012,zhao_2012,nedic_2017,lalitha_2018,mitra_2021,bordignon_2020,bordignon_2021,kayaalp_2021}, agents perform (i) a \emph{local} Bayesian update based on their private observations and then (ii) combine their neighbors' beliefs to obtain their updated beliefs. A common assumption in these {studies} is that agents communicate with \emph{all} of their neighbors at each time instant. While this assumption helps modeling microblogging social media such as Twitter, it falls short when modeling private and personal communication over social media platforms, such as WhatsApp. In many situations, people exchange beliefs with a subset of their contacts. This might happen because, for example, data may arrive at high rates such that agents might not be able to communicate with all their neighbors between two consecutive arrivals. Furthermore, for designing communication-efficient networked systems, sparse interactions between the devices can be preferred. For instance, consider an agent that attempts to receive data from multiple neighbors. These transmissions are likely to collide and to avoid such issues, each neighbor can be given turns to communicate by the receiver agent, similar to a MAC layer protocol. The above observations motivate us to study the social learning problem when agents update their beliefs based on \emph{only one randomly chosen neighbor} at each {time instant}.
	
	{In the social learning literature, the closest work to ours is \cite{shahrampour_2013}}, where the authors have considered symmetric gossip schemes and shown that the agents learn the truth eventually with high probability. In contrast to \cite{shahrampour_2013}, we consider diffusion \cite{Sayed14} algorithms and the communication is not necessarily symmetric. In particular, we assume full-duplex communication at nodes {such that} at time $i$ agent $1$ can receive data from agent $2$ while agent $2$ receives from agent $3$. Note that diffusion {algorithms} with random neighbor selections are included in \cite{6994875, 6994883} for optimization problems rather than social learning which is an inference paradigm. {These works analyze the expected mean-square error of the estimated parameters across the network whereas in social learning the fundamental problem is to show that agents learn the true state-of-nature eventually. Hence, the results from optimization literature are not directly applicable. Furthermore, in comparison with studies on the standard social learning algorithms, a modification of the strong law of large numbers is not enough to show the truth learning in our setup. The main results of our work are listed below:}\vspace*{0.7\baselineskip}
	\begin{itemize}
		\item [(i)] Agents learn the truth eventually with \emph{probability one}.
		\item [(ii)] Despite the decreased amount of communication, the asymptotic rate of learning is the same as the standard social learning algorithms where agents interact with all their neighbors.
		\item [(iii)] For a special case where agents replace their beliefs with the neighbor's belief, we provide a large deviations analysis of log-belief ratios that only uses the marginal distributions of the data across agents, i.e., the result does not depend on any coupling between agents.
	\end{itemize}
	
	\section{Notation}
	
	Random variables are denoted with boldface letters whereas their realizations are denoted with plain letters (e.g., $\bL_i$ and $x_i$). For a collection of random variables, $\sigma(.)$ denotes the smallest $\sigma$-algebra pertaining to the collection. Sets and events are denoted with script-style letters (e.g., $\cA$). $|\cA|$ denotes the cardinality of set $\cA$. For vectors $u$ and $v$, $\langle u, v\rangle$ denotes the inner product between $u$ and $v$; and $\|u\|_1$, $\|u\|$ denote the $\ell_1$ and $\ell_2$ norms of $u$ respectively. All the logarithms are assumed to be natural logarithms. For a graph $G = (V,E)$, $\cN_k$ denotes the neighbor of vertex $k$ and $\deg(k) \triangleq |\cN_k|$ is the degree of vertex $k$.

\section{Problem Formulation}\label{sec:formulation}

We consider \( K \) agents with peer-to-peer communication constrained on a graph topology. These agents aim to infer the true hypothesis \( \theta^\circ\) among a finite set of hypotheses $\Theta$. The belief \( \mu_{k,i} (\theta) \) quantifies the confidence that agent $k$ has at time $i$ in the proposition ``$\theta = \theta^\circ$". The vector $\mu_{k,i}$ lives in a $|\Theta|$-dimensional probability simplex. The agents observe partially informative and private observations, i.e., agent \( k \) observes \( \bxi_{k,i} \) at time \( i \), which is distributed according to the likelihood/distribution \( L_{k} (\bxi_{k,i}| \theta^\circ) \). Agent \( k \) knows its likelihood functions \( L_{k} (\cdot | \theta) \) for all \( \theta \in \Theta \). We assume data is identically and independently distributed (i.i.d.) across time; but is not necessarily independent across agents.

As opposed to prior works where agents receive beliefs from all their neighbors at each time instant \( i \), in this work, each agent randomly selects one neighbor at each time instant, independent from the past and receives information from that neighbor. We have the following assumption regarding the communication topology.

\begin{assumption}
The combination matrix \( A \triangleq [a_{\ell k }] \) is left-stochastic, i.e., $\sum_{\ell \in \cN_k} a_{\ell k} = 1$, where the entry \( a_{\ell k} \) represents the probability that agent \( k \) chooses agent \( \ell \) to communicate with. The network is also strongly-connected, which means there is a path with positive weights between every agent pair \( (k,\ell )\) and there is at least one agent \( k \) with \( a_{kk} > 0\).
\end{assumption}

Compared with earlier works, where \( a_{\ell k} \) represents the weight agent \( k \) assigns to the belief obtained from agent \( \ell \), in this work, this weight represents the probability that agent $k$ chooses and receives information from agent $\ell$. This random selection procedure decreases the number of transmissions made at each time slot, i.e., the number of transmissions decreases by $\frac1 {K}\sum_{k=1}^K \deg(k)$ times on average compared to the standard algorithm.

\subsection{The Algorithm}
In our algorithm, agents update their beliefs with a \emph{local} Bayesian rule to obtain intermediate beliefs as {in standard social learning algorithms \cite{jadbabaie_2012,zhao_2012,nedic_2017,lalitha_2018}}:
\begin{align}\label{eq:dif_adapt_step}
 \bpsi_{k,i} (\theta) &= \frac{L_k(\bxi_{k,i} | \theta)\bmu_{k,i-1} (\theta)}{\sum_{\theta^\prime}L_k(\bxi_{k,i} | \theta^\prime)\bmu_{k,i-1} (\theta^\prime)}. \qquad \text{(Adapt)}
\end{align}
Then, each agent chooses one of their neighbors and updates its intermediate belief by taking a weighted geometric average with the chosen neighbor's belief. Specifically, for agent $k$ and $\ell \in \cN_k$,
\begin{align}\label{eq:dif_combine_step}
   \bmu_{k,i}(\theta) &= \frac{\bpsi_{k,i}(\theta)^\alpha \bpsi_{\ell,i}(\theta)^{\bar\alpha} }{\sum_{\theta'}\bpsi_{k,i}(\theta')^\alpha \bpsi_{\ell,i}(\theta')^{\bar\alpha}} \text{, with prob. $a_{\ell k}$ (Combine)}
\end{align}
where \( \alpha \in [0,1) \) is a confidence weight and $\bar\alpha \triangleq 1-\alpha$. We assume $\bmu_{k,0}(\theta) > 0$ for all $k$ and $\theta$. Observe that $\alpha = 0$ corresponds to the case where the agent replaces its intermediate belief with the chosen neighbor's belief. If we allow $\alpha = 1$, it would correspond to the non-cooperative mode of operation; hence this case is of no interest to our work.

\section{Analysis of the Algorithm}\label{sec:analysis}
{Recall that $\theta^\circ$ is the true hypothesis and let $\theta \in \Theta \setminus \{\theta^{\circ}\}$ denote an arbitrary hypothesis.} We study the evolution of the log-belief ratios:
\begin{equation}
\bS_{k,i} \triangleq \log\frac{\bmu_{k,i}(\theta^{\circ})}{\bmu_{k,i}(\theta)},
\end{equation}
which can be verified to evolve according to
\begin{equation}\label{eqn:iteration}
\bS_{k,i} = \alpha (\bL_{k,i} + \bS_{k,i-1}) + \bar\alpha (\bL_{\ell,i} + \bS_{\ell,i-1}) \text{, with prob. $a_{\ell k}$}
\end{equation}
where
\begin{equation}
\bL_{k,i} \triangleq \log \frac{L_k(\bxi_{k,i} | \theta^\circ)}{L_k(\bxi_{k,i} | \theta)}
\end{equation}
is the log-likelihood ratio (LLR) of the data calculated by agent $k$ at time $i$. An equivalent way to express \eqref{eqn:iteration} is
\begin{equation}\label{eqn:iterate}
\begin{split}
    \bS_{k,i} &= \langle\bS_i,\bw_0^{(k)}\rangle\\
    &= \langle\bL_i+\bS_{i-1}, (\alpha I + \bar\alpha \bA_i)\bw_0^{(k)}\rangle\\
    = \langle&\bL_i, (\alpha I + \bar\alpha \bA_i)\bw_0^{(k)}\rangle+\langle\bS_{i-1},(\alpha I + \bar\alpha \bA_i)\bw_0^{(k)}\rangle
\end{split}
\end{equation}
where $\bL_i \triangleq \begin{bmatrix}\bL_{1,i} \dots \bL_{K,i}\end{bmatrix}^T$, $\bS_i \triangleq \begin{bmatrix}\bS_{1,i} \dots \bS_{K,i}\end{bmatrix}^T$, $\bw_0^{(k)} \triangleq \begin{bmatrix}0 \dots 1 \dots 0\end{bmatrix}^T$ is an all-zero vector with a single $1$ at its $k^{\text{th}}$ element; and $\bA_i \triangleq [\ba_{\ell k,i}]$ is a random matrix such that $\ba_{\ell k,i} = 1$ if node $k$ chooses to communicate with node $\ell$ at time $i$. Note that $\sum_{\ell}\ba_{\ell k,i} = 1$ surely, i.e., is a left-stochastic matrix, and $\E[\bA_i] = A$ for all $i$. Furthermore, since each node selects its neighbor identically and independently across time, $\bA_1,\dots,\bA_i$ are i.i.d.

Now, define 
\begin{equation}\bw_{n}^{(k)} \triangleq \prod_{j = i-n+1}^i (\alpha I + \bar\alpha \bA_j)\bw_0^{(k)}
\end{equation}
for $1 \leq n \leq i$. Note that $\bw_{n}^{(k)}$ is a probability vector. We iterate \eqref{eqn:iterate} to obtain
\begin{equation}
\bS_{k,i} = \sum_{n = 1}^{i} \langle \bw_{n}^{(k)},\bL_{i-n+1}\rangle.
\end{equation}
Our aim is to show that (with a.s. denoting almost sure convergence)
\begin{equation}\label{eqn:desired}
\frac 1 i \bS_{k,i} \asto \langle \pi,d \rangle ,\end{equation}where $d$ is the divergence vector with its $k^{\text{th}}$ element being the KL divergence $D_{\text{KL}}(L_k(.|\theta^{\circ}) || L_k(.|\theta))$; and $\pi$ is the Perron vector of $A$. Recall that $\iprod{\pi}{d}$, the asymptotic rate of convergence, is the same as the standard algorithm where agents benefit from all neighbors at each time instant \cite{nedic_2017,lalitha_2018}. As an initial step to prove \eqref{eqn:desired}, we first establish the following result.
\begin{lemma}For all $k$,
\begin{equation}\label{eqn:lemma_1}
\frac 1 i \sum_{n = 1}^{i} \bw_{n}^{(k)} \asto \pi.
\end{equation}
\end{lemma}
\begin{proof}The statement does not depend on $k$, so without loss of generality we take $k=1$ and omit all the superscripts $(k)$. We first show convergence in probability. From Markov's inequality, we have
\begin{equation}\label{eqn:norm}
\bbP\Bigg( \bigg\lVert \frac 1 i \sum_{n = 1}^{i} \bw_n-\pi \bigg\rVert^2 > \epsilon\Bigg) \leq \frac {\E\bigg[\Big\lVert \sum_{n = 1}^{i} (\bw_n-\pi) \Big\rVert^2\bigg]}{\epsilon i^2}.    
\end{equation}
The expected norm on the right-hand side of \eqref{eqn:norm} is equal to
\begin{multline}\label{eq:up1}
        \E\bigg[ \sum_{n = 1}^{i}\lVert\bw_n-\pi \rVert^2\bigg] + 2 \E\bigg[\sum_{m < n} \langle \bw_m-\pi , \bw_n-\pi \rangle\bigg].
\end{multline}
Define for $m\leq i$, $\cF_{m,i} \triangleq \sigma\big(\{\bA_j\}_{i-m+1\leq j \leq i}\big)$
and observe
\begin{equation}\label{eqn:sub2}
\begin{split}
    \E[\langle \bw_m - \pi, \bw_n- \pi \rangle] &= \E\big[\E[\langle \bw_m - \pi, \bw_n- \pi \rangle|\cF_{{m},i}]\big] \\
    &= \E\big[\langle\bw_m- \pi,\E[\bw_n - \pi|\cF_{m,i}] \rangle\big.]\\
    \end{split}
    \end{equation}
Furthermore, $\bw_n = (\alpha I + \bar \alpha \bA_{i-n+1}) \dots (\alpha I + \bar \alpha \bA_{i-m})\bw_m$ and since $\bA_n$'s are i.i.d., 
\begin{equation}\label{eqn:sub1}
    \begin{split}
        \E[\bw_n|\cF_{m,i}] &= (\alpha I + \bar\alpha A)^{n-m}\bw_m.
    \end{split}
\end{equation}
Substituting \eqref{eqn:sub1} into \eqref{eqn:sub2}, we obtain
    \begin{equation}\label{eqn:mixing}
\begin{split}
    \E[\langle \bw_m& -\pi, \bw_n- \pi \rangle]\\
    &= \E\big[\langle\bw_m-\pi, (\alpha I + \bar\alpha A)^{n-m}(\bw_m-\pi)\rangle\big]\\
    &\stackrel{(a)}{\leq}  \E\big[\lVert\bw_m-\pi\rVert\lVert(\alpha I + \bar\alpha A)^{n-m}(\bw_m-\pi)\rVert\big]
\end{split}
\end{equation}
where $(a)$ follows from Cauchy-Schwarz inequality. For a strongly-connected $A$ and for any $\alpha \in [0,1)$, it is known that there exists a $\rho < 1$ such that $\lVert(\alpha I + \bar\alpha A)^{n-m}(\bw_m-\pi)\rVert \leq C(\rho)\rho^{n-m}$, where $C(\rho)$ is a constant that only depends on $\rho$ \cite[Chapter 4]{LevinPeresWilmer2006}. Hence, \eqref{eqn:mixing} is further upper bounded by
\begin{equation}
\E[\lVert\bw_m-\pi\rVert]C\rho^{n-m}.
\end{equation}
Also note that 
\begin{equation}\label{eqn:triangle_eq}
    \E\big[\lVert\bw_m -\pi \rVert\big] \leq \E\big[\lVert\bw_m\rVert\big]+ \lVert\pi \rVert \leq 2
\end{equation}  
and
\begin{equation}\label{eqn:parallelogram}
    \E\big[\lVert\bw_m -\pi \rVert^2\big] \leq 2(\E\big[\lVert\bw_m\rVert^2\big]+ \lVert\pi \rVert^2) \leq 4.
\end{equation}
Using \eqref{eqn:triangle_eq}, \eqref{eqn:parallelogram} we upper bound \eqref{eq:up1}, and consequently upper bound \eqref{eqn:norm} as
\begin{equation}
\bbP\Bigg( \bigg\lVert \frac 1 i \sum_{n = 1}^{i} \bw_n-\pi \bigg\rVert^2 > \epsilon\Bigg) \leq \frac{4 + 4C\rho/(1-\rho)}{\epsilon i}.
\end{equation}
This shows that $\tfrac 1 i \sum_{n = 1}^{i} \bw_n^{(k)} \to \pi$ in probability. Also note that for a $K$-dimensional vector $u$, $\lVert u\rVert_1 \leq \lVert u\rVert \sqrt{K}$. Therefore
\begin{equation}\label{eqn:Borel}
\bbP\Bigg( \bigg\lVert \frac 1 i \sum_{n = 1}^{i} \bw_n-\pi \bigg\rVert_1 > \epsilon\Bigg) \leq \frac{4\sqrt{K} + 2\sqrt{K}\rho/(1-\rho)}{\epsilon i}.
\end{equation}
The last step of the proof follows by a standard trick to obtain the strong law of large numbers \cite[Chapter 7]{grimmett2001probability}. Let $\bz_i \triangleq  \big\|\sum_{n = 1}^{i} (\bw_n-\pi)\big\|_1$. From Borel-Cantelli lemma \cite[Chapter 2]{Martingales}, the subsequence $\bz_{i^2}/i^2 \asto 0$ --- replace $i$ with $i^2$ in \eqref{eqn:Borel} and observe that the right-hand side is summable. Moreover, observe for all $i^2 \leq m \leq (i+1)^2$:
\begin{equation}
\frac{\bz_m}{m} \stackrel{(a)}{\leq} \frac {\bz_{i^2} + \sum_{n=i^2}^m \|\bw_n-\pi \|_1}{i^2} \stackrel{(b)}{\leq} \frac {\bz_{i^2} + {2}\big((i+1)^2-i^2\big)}{i^2}.
\end{equation}
To obtain $(a)$ we upper bounded the numerator with triangle inequality and lower bounded the denominator by using $m \geq i^2$. $(b)$ holds because $\|w-\pi\|_1 \leq {2}$ for any $w$, and $m \leq (i+1)^2$. Since $\frac{\bz_{i^2}}{i^2} \asto 0$ so does $\frac{\bz_{m}}{m}$.
\end{proof}
Taking the inner product of both sides in \eqref{eqn:lemma_1} with $d$, we obtain
\begin{corollary}
$\tfrac 1 i \sum_{n=1}^i \iprod{\bw_n^{(k)}}{d} \asto \iprod{\pi}{d\,}$ for all $k$.
\end{corollary}
Lemma 1 suggests that the convergence results will not depend on $k$. Hence, we assume $k = 1$ and omit all the superscripts $(k)$ if not needed. The next step is to show that \eqref{eqn:desired} holds under a square-integrability assumption on divergences. More precisely,
\begin{lemma}\label{lem:martingale}
Suppose $\E[(\bL_{k,i})^2] < \infty$ for all $k$. Then \eqref{eqn:desired} holds.
\end{lemma}
\begin{proof}
It is sufficient to show that 
\begin{equation}
 \frac 1 i \sum_{n = 1}^{i} \langle \bw_{n},\bL_{i-n+1}-d\rangle \asto 0.
\end{equation}
Note that $\bL_i$ is independent of $\bw_j$ for all $j$; and of all $\bL_j$ except itself. The same holds for the $\bA_{i}$'s as well. Therefore, for convenience, let us perform an index change $i-n+1 \mapsto n$ on $(\bL_i, \bA_{i})$ to have the equivalent statement
\begin{equation}\label{eqn:index_changed}
   \frac 1 i \sum_{n = 1}^{i} \langle \tilde \bw_{n},\bL_{n}-d\rangle \asto 0
\end{equation}
where $\tilde \bw_n \triangleq (\alpha I + \bar\alpha \bA_n) \tilde \bw_{n-1}$ and $\tilde \bw _0= [1,\dots ,0]^T$ according to the index change described above. Kronecker's lemma \cite[Chapter 12]{Martingales} implies that it is sufficient to check if
\begin{equation}\label{eqn:kronecker}
    M_i \triangleq \sum_{n=1}^i \frac 1 n \iprod{\tilde \bw_{n}}{\bL_{n}-d} \text{ converges a.s.}
\end{equation}
Observe that $M_i$ is a martingale with respect to the filtration $\{\tilde\cF_i$\}, where $\tilde\cF_i \triangleq \sigma\big((\bA_n,\bL_n)_{n \leq i}\big)$. This is because
\begin{equation}
\begin{split}
    \E[M_{i+1}|\tilde\cF_i] &= M_i + \tfrac{1}{i+1}\E[\iprod{\tilde \bw_{i+1}}{\bL_{i+1}-d}|\tilde\cF_i]\\
    &\stackrel{(a)}{=} M_i + \tfrac{1}{i+1}\iprod{\E[\tilde \bw_{i+1}|\tilde\cF_i]}{\E[\bL_{i+1}-d|\tilde\cF_i]}\\
    &\stackrel{(b)}{=} M_i\raisetag{10pt}
\end{split}
\end{equation}
where $(a)$ from conditional independence of $\bL_{i+1}$ and $\tilde\bw_{i+1}$ and $(b)$ follows from $\E[\bL_{i+1}-d|\tilde\cF_i]$ being equal to the all-zero vector. Furthermore, $M_i$ is a square-integrable martingale as
\begin{equation}\label{eqn:sq_integrable}
\begin{split}
    \sup_i \E[M_i^2] &= \sum_{n=1}^\infty \frac {\E[\iprod{\tilde\bw_n}{\bL_n {-d}}^2]}{n^2}\\
    &\stackrel{(a)}{\leq} \sum_{n=1}^\infty \frac {\E[\lVert\tilde\bw_n\rVert^2\lVert\bL_n{-d}\rVert^2]}{n^2} < \infty
\end{split}
\end{equation}
where $(a)$ follows from Cauchy-Schwarz inequality. The final expression is bounded since $\lVert\tilde\bw_n\rVert^2 \leq 1$ and $\E[\lVert\bL_n\rVert^2] < \infty $ by our assumption. The above allows the use of Martingale Convergence Theorem \cite[Chapter 11]{Martingales} and therefore $M_i$ a.s. converges. The proof is complete since we have shown \eqref{eqn:kronecker}.
\end{proof}
Our final aim is to relax the square-integrability condition $\E[(\bL_{k,i})^2] < \infty$ to $\E[|\bL_{k,i}|]<\infty$. We show a sufficient condition for the latter. Note that there exists a $C' > 0$ such that $|x| \leq C'e^{-x} + x$ for all $x$. Hence $\E[|\bL_{k,i}|] \leq C'\E[e^{-\bL_{k,i}}] + \E[\bL_{k,i}] = C' + D(L_k(.|\theta^{\circ}) || L_k(.|\theta))$. Therefore, if all the elements of $d$ are finite, this implies $\E[|\bL_{k,i}|]<\infty$ for all $k$.

To extend the result \eqref{eqn:desired} under absolute integrability, we use the following lemma.
\begin{lemma}[{Kolmogorov's Truncation Lemma \cite[Chapter 12]{Martingales}}]\label{lem:kolmogorov}
Consider i.i.d. $X_{1}, X_{2}, \ldots$ where $\E[|X_1|]<\infty$. Let $\mu\triangleq \E[X_1]$ be the common mean. Define
\begin{equation}\label{eqn:truncation}
Y_{i}\triangleq\begin{cases}
X_i,&|X_i| \leq i\\
0, &|X_i| > i
\end{cases}.
\end{equation}
Then
(i) $\mu_i \triangleq \E[Y_i] \rightarrow \mu$;
(ii) $(Y_i-X_i) \asto 0$;
(iii) $\sum i^{-2}\E[Y_i^2]<\infty$.
\end{lemma}

We truncate the vector $\bL_i$ elementwise and obtain $\bK_i$, i.e., we relate $\bL_{k,i}$ to $\bK_{k,i}$ as in \eqref{eqn:truncation}. Let $d_i \triangleq \E[\bK_{i}]$ and repeat the same steps in the proof of Lemma \ref{lem:martingale} with $(\bK_{i},d_i)$ instead of $(\bL_{i},d)$. Observe that the martingale $M_i$ is square-integrable because the sum in \eqref{eqn:sq_integrable} is finite according to (iii) of Lemma 3. Then we have
\begin{equation}
   \frac 1 i \sum_{n = 1}^{i} \langle \tilde\bw_{n},\bK_{n}-d\rangle \asto 0.
\end{equation}
Finally, (i) and (ii) in Lemma \ref{lem:kolmogorov} together imply
\begin{equation}
    \frac 1 i \sum_{n = 1}^{i} \langle \tilde\bw_{n},\bL_{n}-d\rangle \asto 0,
\end{equation}
which yields the following conclusion.
\begin{theorem}[Asymptotic Convergence Rate]\label{thm:conv_rate}
Suppose all elements of $d$ are finite. Then for all $k$,
\begin{equation}\tfrac 1 i \bS_{k,i} \asto \langle\pi,d\rangle.\end{equation}
\end{theorem}
{From Theorem 1, it is immediate that all agents learn the truth eventually if $\langle\pi,d\rangle > 0$. Since all elements of $\pi$ are positive, this condition --- also called global identifiability --- holds if at least one element of $d$ is strictly positive.}
\begin{corollary}[Truth Learning]\label{cor:truth_learning}
Suppose all elements of $d$ are finite and at least one element of $d$ is strictly positive. Then for all $k$,
\begin{equation}
    \bmu_{k,i}(\theta^\circ) \asto 1.
\end{equation}
\end{corollary}
We conclude this section by noting a straightforward extension of our result. Allow each agent to have a different confidence weight $\alpha_k$. Extensions of Theorem \ref{thm:conv_rate} and Corollary \ref{cor:truth_learning} can be obtained as follows: Define the diagonal matrix $J \triangleq [j_{\ell k}]$ where $j_{kk} = \alpha_k$; and 0 otherwise. Replace $(\alpha I + \bar\alpha \bA_n)$ with $(J + \bA_n(I-J))$ and $\pi$ with $\tilde \pi$, the Perron vector of $(J + A(I-J))$. By following the same steps as above, the extension will be immediate.
\section{A Special Case: Replacement}\label{sec:replacement}
In this section we study the special case $\alpha = 0$, where agents replace their beliefs with their neighbors'. At first sight, it is not intuitive whether all agents would learn the truth eventually. This is because the truthful beliefs might be lost upon replacement. However, all the results of Section \ref{sec:analysis} hold for this case as well. Furthermore, when agents replace their beliefs, Theorem \ref{thm:conv_rate} has a much shorter proof, which we provide below.
\begin{proof}[Short Proof of Theorem \ref{thm:conv_rate} for $\alpha = 0$]Again assume $k = 1$.
Observe in this case $\tilde \bw_n$'s --- recall the index change $i-n+1 \mapsto n$ in Lemma 2 --- are vectors with their $\ell_n^{\text{th}}$ element being one for some $\ell_n$ and the others being zero. Now consider a Markov chain governed by the transition kernel $A$, with the state space $\{1,\dots,K\}$ and observe that the random variable $\bell_n$ is the state at time $n$ with the initial state being $1$. Now, write down $\frac 1 i \bS_i$ as
\begin{equation}\label{eqn:simple}
\begin{split}
    \frac 1 i \bS_i &= \frac 1 i \sum_{n = 1}^{i} \langle \tilde \bw_{n},\bL_{n}\rangle = \sum_{k= 1}^K \frac 1 i\sum_{n:\bell_n = k}\bL_{k,n}\\
    &= \sum_{k= 1}^K \frac {\bmm_i^{(k)}} i \frac 1 {\bmm_i^{(k)}}\sum_{n:\bell_n = k}\bL_{k,n}
\end{split}
\end{equation}
where ${\bmm_i^{(k)}}$ is the number of visits to state $k$ in $i$ transitions. Since the Markov chain governed by $A$ is communicating and aperiodic, it is known from a standard result in Renewal theory \cite[Chapter 3]{ross1996stochastic} that $\frac {\bmm_i^{(k)}} i \asto \pi_{k}$. This also implies $\bmm_i^{(k)}\asto \infty$ and therefore $\frac 1 {\bmm_i^{(k)}}\sum_{n:\bell_n =k}\bL_{k,n} \asto d_{k}$ by the strong law of large numbers. Combining these with \eqref{eqn:simple}, we see \eqref{eqn:desired} holds.
\end{proof}

\begin{figure*}[hbt!]
	\centering
	\subfloat[a][]{
  \includegraphics[width=.2\linewidth]{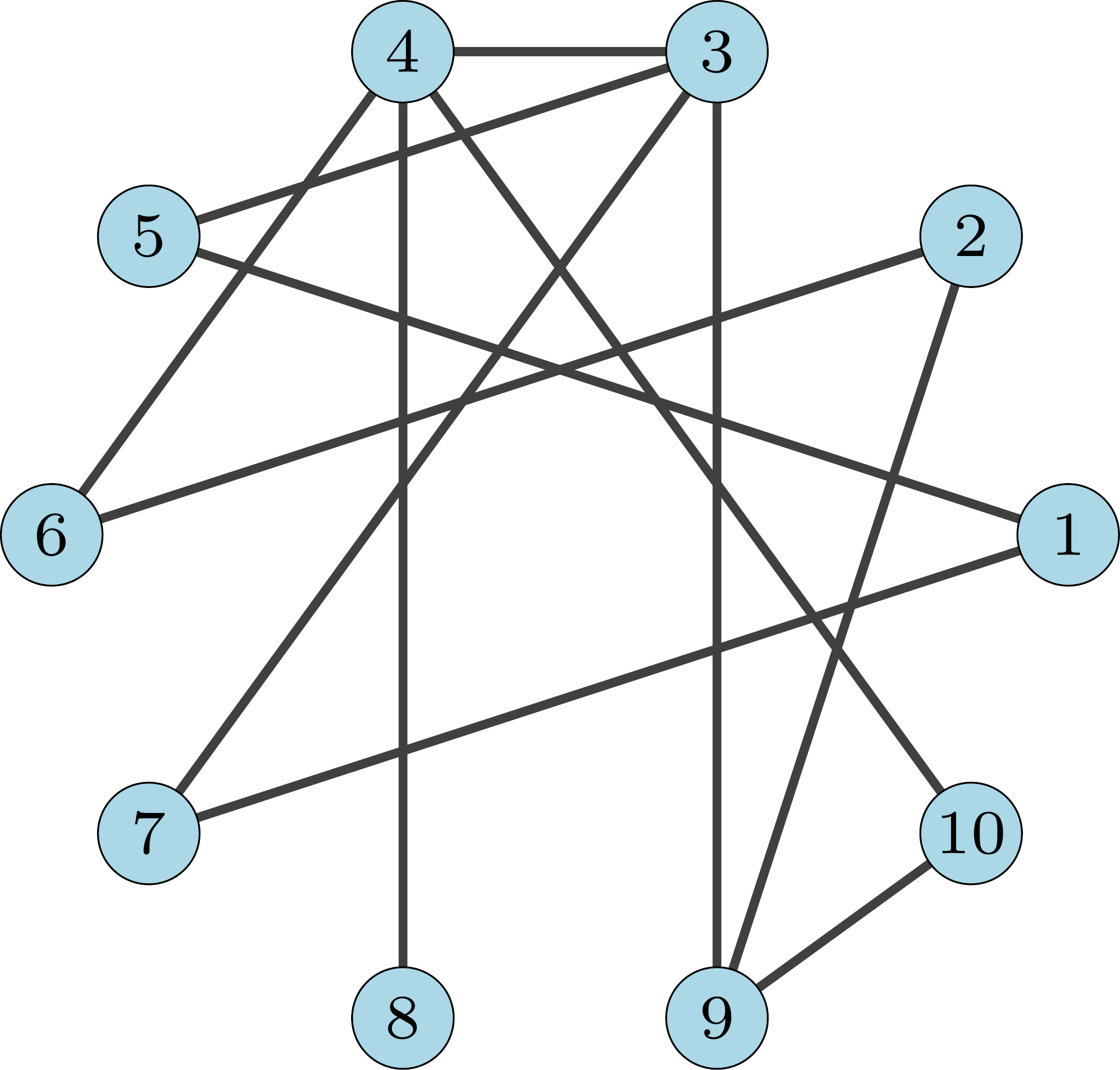}\label{fig:original network}\label{fig:network}
  }\hfil
  \subfloat[b][]{
  \includegraphics[width=.25\linewidth]{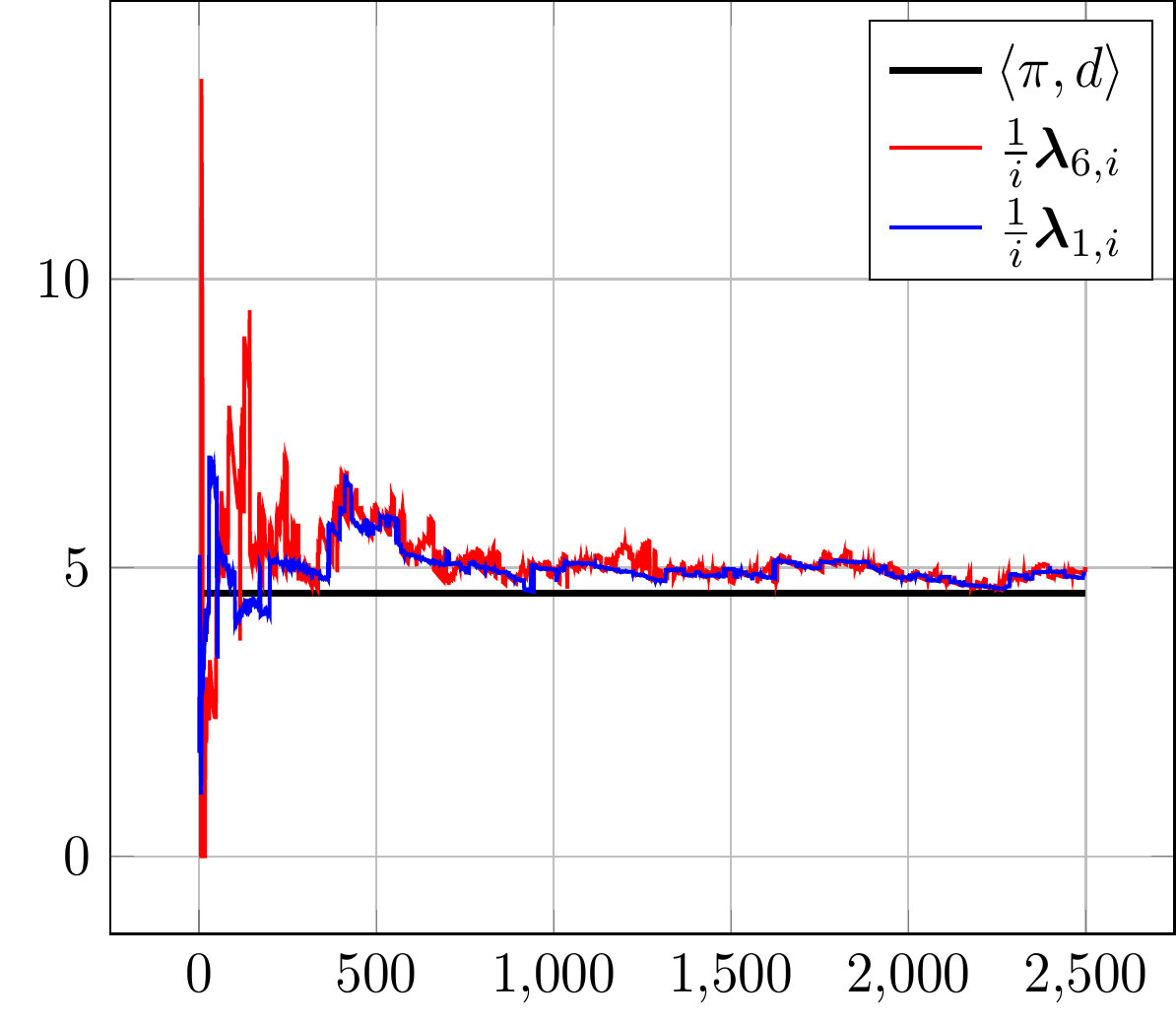}\label{fig:sample_paths}
  }\hfil
  \subfloat[c][]{
  \includegraphics[width=.25\linewidth]{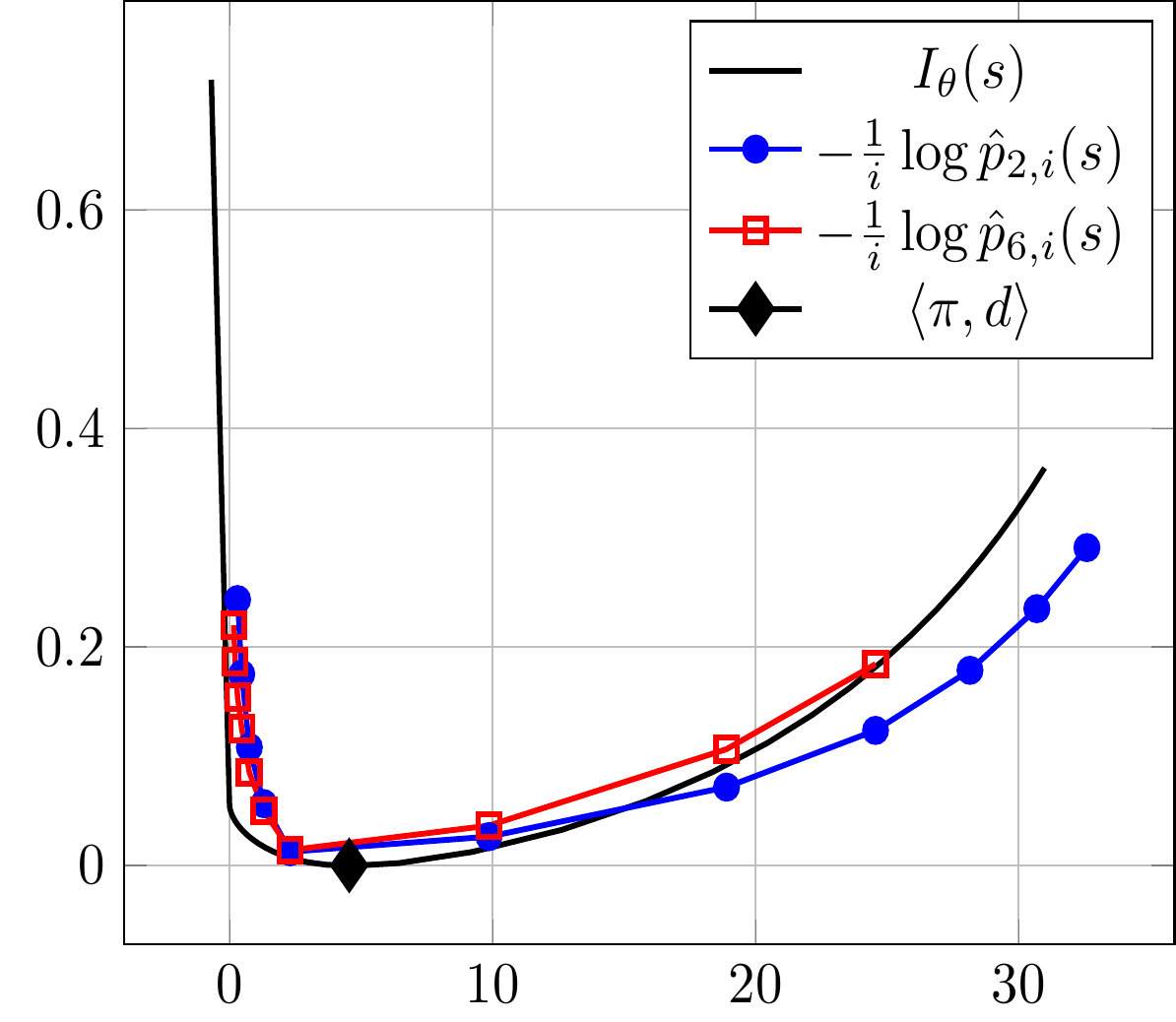}\label{fig:rate_fnc}
  }\caption{(a) The network with $K=10$ nodes and 12 edges. (b) Sample paths of $\frac 1 i \bS_{1,i}$ and $\frac 1 i \bS_{6,i}$, shown by blue and red curves respectively. The asymptotical convergence rate $\iprod{\pi}{d}$ is drawn as a straight horizontal line at 4.55. (c) Large deviation estimates. $I_\theta$ is found numerically and drawn as the black curve; which touches zero exactly at $\iprod{\pi}{d}$. The Monte-Carlo estimates for $\hat p_{2,i}(s)$ and $\hat p_{6,i}(s)$ are drawn as blue and red curves respectively.}
\end{figure*}

Suppose the truth learning terminates at step $i$ for a large $i$. We now shift our attention to the event that a node makes an error upon termination. More precisely, we are interested in evaluating the probability of this event. To this end, we emphasize an important observation from the proof above: $\tilde \bw_i$ evolves according to a finite-state Markov chain. Moreover, $\frac 1 i \bS_i$ can be viewed as the average reward of a Markov Reward Process. Knowing the underlying dependence structure of $\tilde \bw_i$'s allows us to invoke the known results in Large Deviations Principle (LDP). We have a special case of Gärtner-Ellis theorem that implies --- see also \cite{7576662,7491269}:
\begin{theorem}[Theorem 3.1.2 in \cite{dembo2009large}]\label{thm:large_dev}
Set $A(t) \triangleq \big[a_{\ell k}\E[e^{t\bL_{\ell}}]\big]$ and let $\Lambda(t)$ be the Perron-Frobenius eigenvalue of $A(t)$. Then for any $\Gamma \subseteq \mathbb{R}$,
\begin{align}
    -\inf_{s\in \Gamma^\circ} I(s) &\leq \liminf_{i \to \infty} \frac 1 i \log \bbP\Big(\frac 1 i \bS_i \in \Gamma\Big)\label{eqn:ldp1}\\
    &\leq \limsup_{i \to \infty} \frac 1 i \log\bbP\Big(\frac 1 i \bS_i \in \Gamma\Big) \leq -\inf_{s\in \overline{\Gamma}} I(s)\label{eqn:ldp2}
\end{align}
where $I(s)\triangleq \sup_{t\in \mathbb{R}} st - \log \Lambda(t)$ is the Legendre-Fenchel transform of $\log\Lambda(t)$ and $\Gamma^\circ$, $\overline{\Gamma}$ denote interior and closure of $\Gamma$ respectively.
\end{theorem}

Note that Theorem \ref{thm:large_dev} only requires knowledge the marginals, i.e., $L_k(.|\theta)$; and $\E[e^{t\bL_{\ell}}]$ to be finite. Therefore, without any knowledge of the joint distribution of the data, the rate function $I(s)$ can be calculated. This allows to approximate the error probabilities of agents given their respective decision rules. For instance, suppose node $k$ decides based on a maximum-likelihood rule, i.e., believes $\argmax_{\theta}\bmu_{k,i}(\theta)$. Let $\cE_{k,i}$ denote the event that the node $k$ makes an error at time $i$, whose probability at $i\to\infty$ can be lower and upper bounded as
\begin{align}
    -\max_{\theta\neq\theta^\circ}\inf_{s < 0} I_\theta(s) &\leq \liminf_{i \to \infty} \frac 1 i \log \bbP(\cE_{k,i})\\
    &\leq \limsup_{i \to \infty} \frac 1 i \log \bbP(\cE_{k,i}) \leq -\max_{\theta\neq\theta^\circ}\inf_{s\leq 0} I_{\theta}(s)
\end{align}
where $I_{\theta}(s)$ is the rate function corresponding to $\bS_{k,i}(\theta)$. The bounds above do not depend on $k$.

	\section{Numerical Results}
	In this section we present numerical results based on the simulations performed over the network of $K= 10$ nodes in Figure \ref{fig:network}. We set $\alpha = 0$, i.e., we simulate the replacement algorithm of Section \ref{sec:replacement}, and aim to solve a binary hypothesis testing problem between $\theta^\circ$ and $\theta$. For simplicity, we assume that the data is independent across agents --- note that the rate function $I_\theta(s)$ is unaffected by such assumption. Moreover, node $k$ observes a Gaussian random variable with unit variance under each hypothesis and with zero mean under $\theta^\circ$; and with mean $\nu_k$ under $\theta$. We have set $\nu \triangleq [3, 8, 0, 0, 3, 0, 3, 0, 0, 0]$, so for instance $\nu_5 = 3$. Observe that $\theta^\circ$ can only be identified by nodes 1,2,5 and 7. The combination matrix $A$ is chosen according to a lazy Metropolis rule \cite{Metropolis1953}, namely we set $B \triangleq [b_{\ell k}]$ with $b_{\ell k} = \max\{\deg(\ell),\deg(k)\}^{-1}$ for $\ell \neq k$ and $b_{\ell \ell} = 1-\sum_{k\neq\ell}b_{\ell k}$. Then we set $A = \tfrac 1 2 I + \tfrac 1 2 B$.
	
	$A$ is symmetric, hence doubly stochastic; which implies $\pi_{k} = \tfrac 1 K$ for all $k$. Furthermore, we can straightforwardly calculate $d_k = E\Big[\frac{L_k(\bxi_{k}|\theta^\circ)}{L_k(\bxi_k|\theta)}\Big] = \tfrac 1 2 \nu_k^2$; which gives the asymptotic convergence rate as $\iprod{\pi}{d} = 4.55$. Figure \ref{fig:sample_paths} shows two sample paths corresponding to $\tfrac 1 i \bS_{1,i} $ and $\tfrac 1 i \bS_{6,i}$ for $i \leq 2500$; and is consistent with our theoretical results from Section \ref{sec:analysis} as the paths seem to converge to 4.55.
	

In Figure \ref{fig:rate_fnc}, we have drawn the rate function $I_\theta(s)$ as a black solid line. Note that $I_\theta(s)$ touches zero exactly at $s= 4.55$, indicated with the solid diamond. We have also obtained Monte Carlo estimates of deviations $p_{k,i}(s) \triangleq \bbP(\frac 1 i \bS_{k,i} > s)$ for $s > 4.55$ and $\bbP(\frac 1 i \bS_{k,i} < s)$ for $s < 4.55$; $k = 2,6$, $i = 2500$ to check if they fit in with $I_\theta(s)$. However, the rate function suggests that one should expect $p_{k,i}(s)$ to become exponentially small with $i$. Hence, standard Monte Carlo method requires that the number of experiments should be exponentially large with $i$, which is impractical. We therefore resort to an importance sampling method by using a tilted Markov chain and tilted Gaussians, where the tiltings depend on $s$. We omit the details of the tilting procedure, and refer the reader to \cite{Collamore, Nummelin} for details. Denote the tilted measure as $\bQ_s$. Then our Monte Carlo estimate {(for $s < 4.55$)} is 
\begin{equation}
\hat p_{k,i}(s) \triangleq \frac 1 N \sum_{n = 1}^N \frac{d\bbP^{(n)}}{d\bQ_s^{(n)}}\mathbbm{1}\{\lambda_{k,i}^{(n)} < s\}
\end{equation}
where the superscript $(n)$ denotes the $n^\text{th}$ realization under the tilted measure $\bQ_s$. $\frac{d\bbP^{(n)}}{d\bQ_s^{(n)}}$ is a measure change variable, i.e., Radon--Nikodym derivative, which turns out to be expressed as a product $\prod_{j = 1}^i f(x_j^{(n)}, w_j^{(n)})$ for some $f$. Note that $\E_{\bQ_s}[\hat p_{k,i}(s)] = p_{k,i}(s)$. {For $s > 4.55$, we replace the indicator function with $\mathbbm{1}\{\lambda_{k,i}^{(n)} > s\}$}. We performed $N = 60$ experiments to obtain each marker in Figure \ref{fig:rate_fnc}.


\section{Discussion}
{Under randomized collaborations, agents still learn the truth, and at the same rate compared to the standard algorithms. Although the asymptotic rate $\lim_{i\to \infty}\frac 1 i \bS_{k,i}$ does not depend on $\alpha$, the statistics of $\bS_{k,i}$ depend on $\alpha$ for finite $i$. For $\alpha = 0$, we provided a finite-time analysis based on large deviation estimates for finite-state Markov chains. However, for $\alpha > 0$, the corresponding Markov chain has a continuous state space and a finite-time analysis requires more advanced machinery.}

%
	
	\bibliographystyle{IEEEtran}
\bibliography{ref}
\end{document}